\documentclass[oribibl,envcountsect,envcountsame]{llncs}

\usepackage{latexsym}
\usepackage{epsfig}
\usepackage{verbatim}
\usepackage{floatflt}
\usepackage{cite}
\usepackage{amssymb}
\usepackage{amsmath}
\usepackage{graphicx}

\newcommand{\ignore}[1]{}

\newtheorem{corol}[theorem]{Corollary}

\DeclareMathOperator{\dec}{dec}
\DeclareMathOperator{\pj}{proj}
\DeclareMathOperator{\lt}{lift}
\newcommand{\Vred}{V_\text{red}}

\newcommand{\C}{\mathcal C}
\newcommand{\U}{\mathcal U}
\newcommand{\R}{\mathbb R}
\newcommand{\N}{\mathbb N}

\newcommand{\eqdef}{:=}

\title{Data Structures on Event Graphs\thanks{A preliminary
version appeared as B. Chazelle and W. Mulzer,
\emph{Data Structures on Event Graphs} in Proc.~20th ESA, pp.~313--324, 2012}}

\date{}

\author{
Bernard Chazelle\inst{1}
\and
Wolfgang Mulzer\inst{2}
}
\institute{
Department of Computer Science, Princeton University, USA
\email{chazelle@cs.princeton.edu}
\and
Institut f{\"ur} Informatik, Freie Universit{\"a}t Berlin, Germany
\email{mulzer@inf.fu-berlin.de}
}

\begin{document} \maketitle

\begin{abstract}
We investigate the behavior of 
data structures when the input
and operations are generated by an
\emph{event graph}. 
This model is inspired by 
Markov chains.
We are given a fixed graph $G$,
whose nodes 
are annotated with operations of the type 
\emph{insert}, \emph{delete}, and \emph{query}.
The algorithm responds to the requests 
as it encounters them during a (random or adversarial) walk in $G$.
We study the  limit behavior  of such a 
walk and give an efficient algorithm for recognizing which structures
can be generated.
We also give a near-optimal algorithm for successor searching if
the event graph is a cycle and the walk is adversarial. For a random
walk, the algorithm becomes optimal. 
\end{abstract}

\section{Introduction}\label{introduction}

In contrast with the traditional adversarial assumption
of worst-case analysis,
many data sources are modeled by Markov chains
(e.g., in queuing, speech, gesture, protein homology,
web searching, etc.). These models are very appealing
because they are  widely applicable and simple
to generate.
Indeed, locality of reference, an essential pillar in the
design of efficient computing systems, is often captured by 
a Markov chain modeling the access distribution.
Hence, it does not come as a surprise that 
this connection has motivated and guided much
of the research on self-organizing data 
structures and online algorithms
in a Markov setting~\cite{Chassaing93,Hotz93,KapoorRe91,KarlinPhRa00,
KonnekerVa81,LamLeSi84,PhatarfodPrDy97,SchulzSc96,ShedlerTu72,VitterKr96}.
That body of work should be seen as part of 
a larger effort to understand
algorithms that exploit the fact that
input distributions often exhibit only a small amount
of entropy. This effort is driven not only by the
hope for improvements in practical applications 
(e.g., exploiting coherence in data streams),
but it is also motivated by theoretical questions: for example, the key 
to resolving the problem of designing an optimal
deterministic algorithm for minimum spanning trees lies in
the discovery of an optimal heap for constant-entropy 
sources~\cite{Chazelle00}.
Markov chains have been studied intensively, and there
exists a huge literature on them (e.g.,~\cite{LevinPeWi09}).
Nonetheless, the focus has been on state functions
(such as stationary distribution or commute/cover/mixing times)
rather than on the behavior of complex objects evolving over them.
This leads to a number of fundamental questions which, we hope,
will inspire further research.

Let us describe our model in more detail.
Our object of interest is a structure $\mathcal{T}(X)$ that evolves
over time. The structure $\mathcal{T}(X)$ is defined over a
finite subset $X$ of a universe $\U$.
In the simplest case, we have $\U = \N$  and 
$\mathcal{T}(X) = X$. This corresponds to the classic
dictionary problem where we need to maintain a subset of a 
given universe. We can also imagine more complicated 
scenarios such as $\U = \R^d$ with 
$\mathcal{T}(X)$ being the Delaunay triangulation of $X$.
An \emph{event graph} $G= (V,E)$ specifies restrictions
on the queries and updates that are applied to $\mathcal{T}(X)$.
For simplicity, we assume that $G$ is undirected and connected.
Each node $v \in V$ is associated with an item $x_v\in \U$
and corresponds to one of three possible requests:
(i) \texttt{insert}$(x_v)$;
(ii) \texttt{delete}$(x_v)$; or
(iii) \texttt{query}$(x_v)$.
Requests are specified by following a walk in $G$,
beginning at a designated start node of $G$ and hopping
from node to neighboring node. We consider both \emph{adversarial}
walks, in which the neighbors can be chosen arbitrarily, and
\emph{random} walks, in which the neighbor is chosen uniformly at random.
The latter case corresponds to the classic Markov chain model. 
Let $v^t$ be the node of $G$ visited at time $t$ and 
let $X^{t}\subseteq \U$ be the set of \emph{active elements},
i.e., the set of items inserted prior to time $t$ 
and not deleted after their last insertions.
We also call $X^{t}$ an \emph{active set}.
For any $t>0$, $X^{t}= X^{t-1}\cup \{x_{v^t}\}$
if the operation at $v^t$ is an insertion and
$X^{t}= X^{t-1}\setminus \{x_{v^t}\}$ in the case of deletion.
The query at $v$ depends on the structure under consideration
(successor, point location, ray shooting, etc.).
Another way to interpret the event graph is as a 
finite automaton that generates words over an alphabet with 
certain cancellation rules.

Markov chains are premised on forgetting the past.
In our model, however, the structure
$\mathcal{T}(X^{t})$ can remember quite a bit. In fact, we can
define a secondary graph over the much larger
vertex set $V\times 2^{\U_{|V}}$, 
where $\U_{|V} = \{x_v|\, v\in V\}$ denotes those 
elements in the universe that occur as labels in $G$, see Fig.~\ref{fig:decorated_ex}. 
We call this larger graph the \emph{decorated graph},
$\dec(G)$, since the way to think of this
secondary graph is to picture each node $v$ of $G$ being
``decorated'' with the subsets $X \subseteq \U_{|V}$. 
(We define the vertex set using $2^{\U_{|V}}$
in order to allow for every possible initial subset $X$.)
Let $n$
be the number of nodes in $G$. Since $|\U_{|V}| \leq n$, an edge $(v,w)$
in the original graph gives rise to up to $2^n$ edges 
$(v,X)(w,Y)$ in the decorated graph, 
with $Y$ derived from $X$ in the obvious way.
A trivial upper bound on the number of states is $n2^n$, which is
essentially tight.
If we could afford to store all of $\dec(G)$,
then any of the operations at the nodes of the
event graph could be precomputed and the running time per step would be constant.
However, the required space might be huge,
so the main question is 

\begin{center}
\emph{Can the decorated graph
be compressed with no loss of performance?}
\end{center}

This seems a difficult question to answer in general. 
In fact, even counting the possible active sets in decorated graphs seems
highly nontrivial, as it reduces to counting words 
in regular languages augmented 
with certain cancellation rules. Hence, in this paper we 
focus on basic properties and special cases that highlight
the interesting behavior of the decorated graph.
Beyond the results themselves, the main contribution
of this work is to draw the attention of algorithm designers 
to a more realistic
input model that breaks away from worst-case analysis.

\paragraph{Our Results.}
The paper has two main parts. In the first part, 
we investigate some basic properties of decorated graphs.
We show that the decorated graph $\dec(G)$ has a unique strongly connected
component that corresponds to the limiting phase of a walk
on the event graph $G$, and we give characterizations
for when a set $X \subseteq \U_{|V}$ appears as an active set 
in this limiting phase. We also show that whether $X$ is such an
active set can be decided in linear time (in the size of $G$).

In the second part, we consider the problem of maintaining a dictionary
that supports successor searches during a one-dimensional
walk on a cycle. We show how to achieve linear space and constant 
expected time for a random walk. If the walk is adversarial,
we can achieve a similar result with near-linear storage.
The former result  is in
the same spirit as previous work by the authors on randomized incremental 
construction (RIC) for Markov sources~\cite{ChazelleMu09}. 
RIC is a fundamental algorithmic paradigm
in computational geometry that uses randomness for the construction
of certain geometric objects, and we showed that there is no 
significant loss of efficiency if the  randomness comes 
from a Markov chain with sufficiently high 
conductance.

\begin{figure}
\begin{center}
\includegraphics{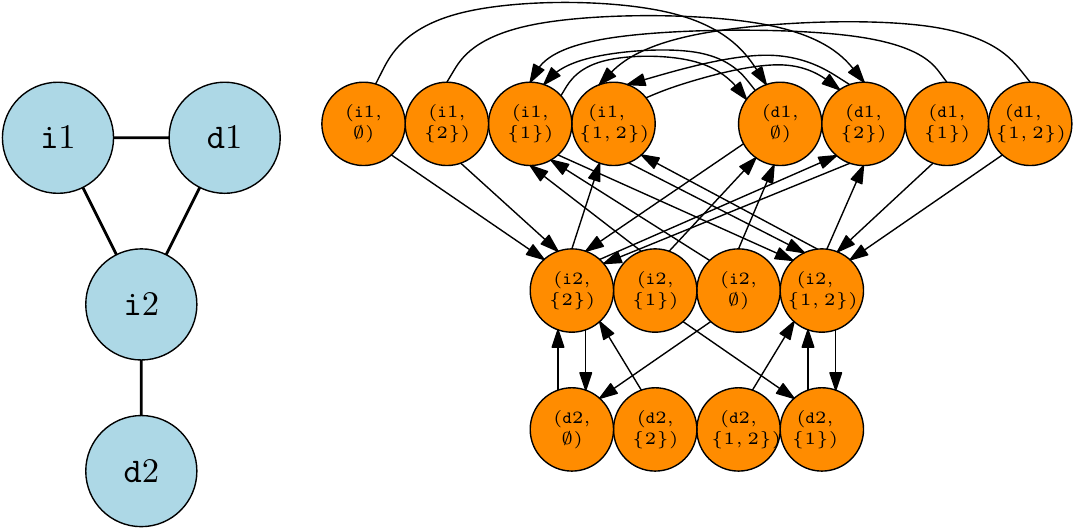}
\end{center}
\caption{An event graph over four vertices and the associated
  decorated graph. Each node of the event graph is replaced by
  four nodes decorated with the subsets of $\{1,2\}$.
}
\label{fig:decorated_ex}
\end{figure}

\section{Basic Properties of Decorated Graphs}

We are given a labeled,
connected, undirected graph $G = (V, E)$. In this section, we 
consider only labels of the form $\texttt{i}x$ and $\texttt{d}x$, 
where $x$ is an element from a finite universe $\U$ and \texttt{i} 
and \texttt{d} stand for \texttt{insert} and 
\texttt{delete}. We imagine an adversary that maintains a subset 
$X \subseteq \U$ while walking on $G$ and performing the corresponding operations
on the nodes. Since the focus of this section is the evolution of
$X$ over time, we ignore queries for now.

Recall that $\U_{|V}$  denotes the elements that appear on the nodes of 
$G$.
For technical convenience, we require that for every 
$x \in \U_{|V}$ there is at least one node labeled \texttt{i}$x$ and
at least one node labeled \texttt{d}$x$. The walk on $G$ is
formalized through the \emph{decorated graph} $\dec(G)$. 
The graph $\dec(G)$ is a directed graph on vertex set 
$V' \eqdef V \times 2^{\U_{|V}}$. The pair $((u, X), (v, Y))$ is an 
edge in $E'$ if and only if $\{u, v\}$ is an edge in $G$ 
and $Y = X \cup \{x_v\}$ or 
$Y = X \backslash \{x_v\}$ depending on whether $v$ is labeled 
\texttt{i}$x_v$ or \texttt{d}$x_v$, see Fig.~\ref{fig:decorated_ex}. 

By a \emph{walk} $W$ in a (directed or undirected) graph, we mean any finite 
sequence of nodes such that the graph contains an edge from each node in
$W$ to its successor in $W$ (in particular, a node may appear multiple 
times in $W$).  Let $A$ be a walk in $\dec(G)$. Recall that the 
nodes in $A$ are tuples, consisting of a node in $G$ and a subset 
of $\U_{|V}$. By taking the first elements of the nodes 
in $A$, we obtain a walk in $G$, the \emph{projection} of $A$, denoted 
by $\pj(A)$. For example, in Fig.~\ref{fig:decorated_ex}, the projection
of the walk 
$(\texttt{i}1, \emptyset), (\texttt{i}2, \{2\}), (\texttt{i}1, \{1,2\}),
(\texttt{d}1, \{2\})$ in the decorated graph is the walk
$\texttt{i}1, \texttt{i}2, \texttt{i}1, \texttt{d}1$ in the event graph. 
Similarly, let $W$ be a walk in $G$ with start node 
$v$, and let $X \subseteq 2^{\U_{|V}}$. Then the \emph{lifting} of $W$ with
respect to $X$ is the walk in $\dec(G)$  that begins at node $(v, X)$ and 
follows the steps of $W$ in $\dec(G)$. We denote this walk by
$\lt(W, X)$.
For example, in Fig.~\ref{fig:decorated_ex}, we have
$\lt((\texttt{i}1, \texttt{i}2, \texttt{i}1, \texttt{d}1), \emptyset) = 
((\texttt{i}1, \emptyset), (\texttt{i}2, \{2\}), (\texttt{i}1, \{1,2\}),
(\texttt{d}1, \{2\}))$.

Since $\dec(G)$ is a directed graph, it can be decomposed into
strongly connected components that induce a directed acyclic
graph $D$. 
We call a strongly connected component of $\dec(G)$ a \emph{sink 
component} (also called essential class in Markov chain theory), 
if it corresponds to a sink 
(i.e., a node with out-degree $0$) in $D$. 
First, we observe that every node of $G$ is represented in each sink
component of $\dec(G)$, see Fig~\ref{fig:decorated_comps}.

\begin{figure}
\begin{center}
\includegraphics{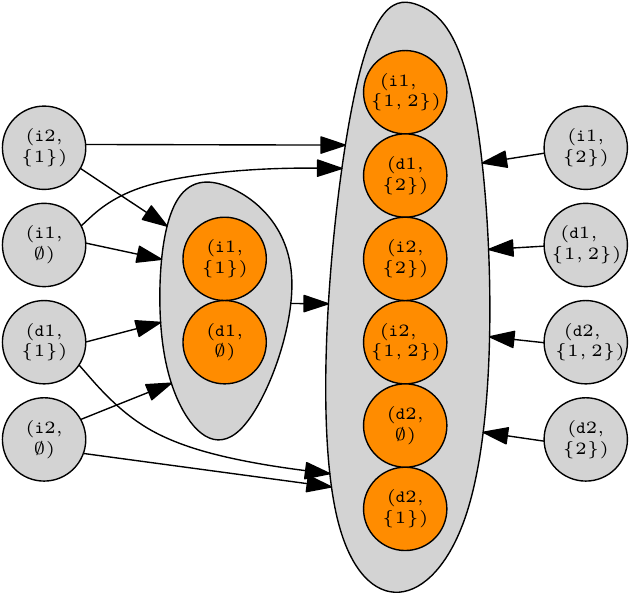}
\end{center}
\caption{The decomposition of the decorated graph from 
  Fig.~\ref{fig:decorated_ex} into strongly connected components.
  There is a unique sink component in which each node from the
  event graph is represented.
}
\label{fig:decorated_comps}
\end{figure}

\begin{lemma}\label{lem:all-first-nodes}
Let $\C$ be a sink component of $\dec(G)$.
For each vertex $v$ of $G$, there exists at least one subset $Y \subseteq \U_{|V}$
such that $(v,Y)$ is a node  in $\C$. In other words, $v$ is the first element
of at least one node in $\C$.
\end{lemma}

\begin{proof}
Let $(w, X)$ be any node in $\C$. Since $G$ is connected, 
there is a walk $W$ in $G$ from
$w$ to $v$, so $\lt(W, X)$ ends in a node in $\C$ whose first element is
$v$.
\qed\end{proof}

Next, we show that to understand
the  behaviour of a walk on $G$ in the limit, it
suffices to focus on a single sink component of $\dec(G)$.

\begin{lemma}\label{lem:uniqueSink}
In $\dec(G)$ there exists a unique sink component $\C$ such that 
for every node $(v, \emptyset)$ in $\dec(G)$, $\C$ is the only
sink component that $(v, \emptyset)$ can reach. 
\end{lemma}

\begin{proof}
Suppose there is a node $v$ in $G$ such that  $(v, \emptyset)$ can reach two 
different sink components $\C$ and $\C'$ in $\dec(G)$. 
By Lemma~\ref{lem:all-first-nodes},
both $\C$ and $\C'$ must contain at least one node with first
element $v$. Call these nodes $(v,X)$ (for $\C$) and
$(v, X')$ (for $\C'$).
Furthermore, by assumption $\dec(G)$ contains a walk $A$ from 
$(v, \emptyset)$ to $(v,X)$ and a walk $A'$ from 
$(v, \emptyset)$ to $(v, X')$. 
Let $W \eqdef \pj(A)$ and $W' \eqdef \pj(A')$. Both $W$ and $W'$ are
closed walks in $G$ that start and end in $v$,
so their concatenations $WW'W$ and $W'W'W$ are
valid walks in $G$, again with start and end vertex $v$.
Consider the lifted walks 
$\lt(WW'W, \emptyset)$ and $\lt(W'W'W, \emptyset)$ in
$\dec(G)$. We claim that these two walks have the same end node $(v, X'')$. 
Indeed, for each $x \in \U_{|V}$, whether $x$ appears in $X''$ or not 
depends solely on whether the label \texttt{i}$x$ or the label \texttt{d}$x$
appears last on the original walk in $G$. This is the same
for both $WW'W$ and $W'W'W$. Hence,  $\C$ and $\C'$ must both contain
$(v, X')$, a contradiction to the assumption that they are distinct sink
components. 
Thus, each node $(v, \emptyset)$ can reach exactly one sink component. 

Now consider two distinct nodes $(v, \emptyset)$ and $(w, \emptyset)$
in $\dec(G)$ and assume that they reach the sink components $\C$ and $\C'$,
respectively. Let $W$ be a walk in $G$ that 
goes from  $v$ to $w$ and let $W' \eqdef \pj(A)$, where $A$ is 
a walk in $\dec(G)$ that 
connects $w$ to $\C'$. Since $G$ is undirected, the reversed walk
$W^R$ is a valid walk in $G$ from $w$ to $v$. 
Now consider the walks $Z_1 \eqdef WW^RWW'$ and $Z_2 \eqdef W^RWW'$. 
The walk $Z_1$ begins in
$v$, the  walk $Z_2$ begins in $w$, and they both have the same end node. 
Furthermore, for each $x \in \U_{|V}$, the label \texttt{i}$x$ appears last
in $Z_1$ if and only if it appears last in $Z_2$.
Hence, the lifted walks $\lt(Z_1, \emptyset)$ and $\lt(Z_2, \emptyset)$ 
have the  same end node in $\dec(G)$, so $\C = \C'$.  The lemma
follows.
\qed\end{proof}

Since the unique sink component $\C$ 
from Lemma~\ref{lem:uniqueSink} represents the
limit behaviour of the set $X$ during a walk in $G$, we will
henceforth focus on this component.
Let us begin with a few properties of $\C$.
First, we characterize the nodes in $\C$.
\begin{lemma}\label{lem:C-characterization}
Let $v$ be a node of $G$ and $X \subseteq \U_{|V}$. We have
$(v, X) \in \C$ if and only if there exists a closed walk $W$ in $G$ 
with the following properties:
\begin{enumerate}
  \item the walk $W$ starts and ends in $v$:
  \item for each $x \in \U_{|V}$, there is at least one node in $W$ with
     label \texttt{i}$x$ or \texttt{d}$x$;
  \item we have $x \in X$ if and only  if the last node in $W$ referring 
  to $x$ is an insertion and $x \not\in X$ if and only if the last node 
	in $W$ referring to $x$ is a deletion.
\end{enumerate}
\end{lemma}

\noindent
We call the walk $W$ from Lemma~\ref{lem:C-characterization}
a \emph{certifying walk} for the node $(v, X)$ of $\C$.
For example, as we can see in Fig.~\ref{fig:decorated_comps},
the sink component of our example graph contains the node
$(\texttt{d}2, \{2\})$. A certifying walk for this node is
$\texttt{d}2, \texttt{i}2, \texttt{d}1,\texttt{i}2, \texttt{d}2$.
\begin{proof}
First, suppose there is a walk with the given properties.
By Lemma~\ref{lem:all-first-nodes}, there is at least one  node 
in $\C$ whose first element is $v$, say $(v, Y)$. 
The properties of $W$ immediately imply that 
the walk $\lt(W, Y)$ ends in $(v, X)$, which proves the 
``if''-direction of the lemma.

Now suppose that $(v,X)$ is a node in $\C$. Since $\C$ is strongly connected,
there exists a closed walk $A$ in $\C$ that starts and ends at $(v,X)$ and
visits every node of $\C$ at least once. Let $W \eqdef \pj(A)$. 
By Lemma~\ref{lem:all-first-nodes} and our assumption on the labels
of $G$, the walk $W$ contains for every element $x \in \U_{|V}$ 
at least one node
with label \texttt{i}$x$ and one node with label \texttt{d}$x$.
Therefore, the walk $W$ meets all the desired properties.
\qed\end{proof}

This characterization of the nodes in $\C$ immediately implies
that the decorated graph can have only one sink component.

\begin{corol}
The component $\C$ is the only sink component of $\dec(G)$.
\end{corol}
\begin{proof}
Let $(v, X)$ be a node in $\dec(G)$. By Lemmas~\ref{lem:all-first-nodes} 
and~\ref{lem:C-characterization}, there exists in $\C$ a node 
of the form $(v, Y)$ and a
corresponding certifying walk $W$. Clearly, the walk $\lt(W, X)$ ends 
in $(v, Y)$. Thus, every node in $\dec(G)$ can reach $\C$, so there
can be no other sink component. 
\qed\end{proof}

Next, we give a bound on the length of certifying walks, from
which we can deduce a bound on the diameter of $\C$.
\begin{theorem}
Let $(v, X)$ be a node of $\C$ and let $W$ be a corresponding certifying walk 
of minimum length. Then $W$ has length at most $O(n^2)$, where $n$ denotes 
the number of nodes in $G$. There are examples where any certifying walk
needs $\Omega(n^2)$ nodes. It follows that $\C$ has diameter $O(n^2)$
and that this is tight.
\end{theorem}

\begin{proof}
Consider the reversed walk $W^{R}$. We subdivide $W^{R}$ into \emph{phases}: 
a new phase starts when $W^{R}$ encounters a node labeled 
\texttt{i}$x$ or \texttt{d}$x$ for an $x \in \U_{|V}$ that it has not seen 
before. Clearly, the number of phases is at most $n$. 
Now consider the $i$-th phase and let $V_i$ be the set of nodes in $G$ 
whose labels refer to the $i$ distinct elements of $\U_{|V}$ that have been
encountered in the first $i$ phases. In phase $i$, the walk $W^{R}$ 
can use only vertices in $V_i$. Since $W$ has minimum cardinality, 
the phase must consist of a shortest walk in $V_i$ from the first node of
phase $i$ to the first node of phase $i+1$.
Hence, each phase consists of at most $n$ vertices and
the length of $W$ is $O(n^2)$.

We now describe the lower bound construction. 
Let $m \geq 2$ be an integer.
The event graph $P$  is a path with $n = 2m + 1$ vertices. The first $m$ 
vertices are labeled
$\texttt{i}m, \texttt{i}(m-1), \dots \texttt{i}1$, in this order. 
The middle vertex is labeled $\texttt{d}m$, and the remaining 
$m$ vertices are labeled $\texttt{d}1, \texttt{d}2, \dots, \texttt{d}m$, in this
order, see Fig.~\ref{fig:large-diameter}. 
Let $v$ be the middle vertex of $P$ and
$\C$ be the unique sink component of $\dec(P)$.
First, note that $(v, X)$ is a node of $\C$ for every 
$X \subseteq \{1, \ldots, m-1\}$. Indeed, given $X \subseteq \{1, \ldots, m-1\}$,
we can construct a certifying walk for $X$ as follows:
we begin at $v$, and for $k = m-1, m-2, \dots, 1$, we walk from $v$ to
$\texttt{i}k$ or $\texttt{d}k$, depending on whether $k$ lies in $X$ or 
not, and back to $v$. This gives a certifying walk for $X$  with 
$2(m-1) + 2(m-2) + \dots + 2 = \Theta(m^2)$
steps. 
Now, we claim that the length of a shortest certifying walk for the node 
$\left(v, 
\left\{2k + 1 \mid k = 0, \ldots, \lfloor m/2 \rfloor - 1 \right\}\right)$
is $\Theta(m^2) = \Theta(n^2)$. Indeed, note that the set 
$ Y = \left\{2k + 1 \mid k = 0, \ldots, \lfloor m/2 \rfloor - 1 \right\}$ 
contains exactly the odd numbers between $1$ and $m-1$. Thus, a certifying 
walk for $Y$ must visit the node $\texttt{i}1$ after all visits to
node $\texttt{d}1$, the node $\texttt{d}2$ after all visits to 
$\texttt{i}2$, etc. Furthermore, the structure of $P$ dictates that 
any certifying walk performs these visits in order from largest 
to smallest, i.e., first comes the last visit to the node for
$m-1$, then the last visit to the node for $m-2$, etc. 
To see this, suppose that there exist $i < j$
such that the last visit to the node for $i$, $w_i$, comes before 
the last visit to the node for $j$, $w_j$. Then the parity of $i$ 
and $j$ must differ, because otherwise the walk must cross $w_i$ 
on the way from $w_j$ to $v$. However, in this case, on the way 
from $w_j$ to $v$, the certfying walk has to cross the node with the 
wrong label for $i$ (\texttt{insert} instead of \texttt{delete}, or vice 
versa), and hence it could not be a certifying walk. It follows 
that any certifying walk for $(v,Y)$ has length $\Omega(n^2)$.

\begin{figure}
\begin{center}
\includegraphics[scale=0.85]{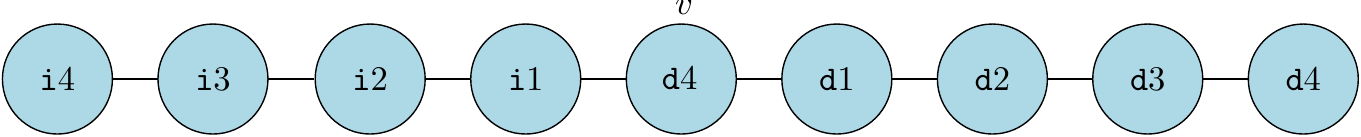}
\end{center}
\caption{The lower bound example for $m = 4$. The
shortest certifying walk for $(v, \{1,3\})$ goes from $v$ to
$\texttt{i}3$, then to $\texttt{d}2$, then to $\texttt{i}1$, and then back to
$v$.}
\label{fig:large-diameter}
\end{figure}

We now show that any two nodes in $\C$ are connected by a walk of
length $O(n^2)$. 
Let $(u, X)$ and  $(v, Y)$ be two such nodes and let $Q$ be a 
shortest walk from $u$ to $v$ in $G$ and
$W$ be a certifying walk for $(v, Y)$. Then $\lt(QW, X)$ is a walk of length
$O(n^2)$ in $\C$ from $(u, X)$ to $(v, Y)$. Hence, the diameter of
$\C$ is $O(n^2)$. Again, the lower bound example from the previous
paragraph applies: 
the length of a shortest walk in $\C$ between $(v, \emptyset)$ and 
$\left(v, 
\left\{2k +1 \mid k = 0, \ldots, \lfloor m/2 \rfloor - 1 \right\}\right)$ 
is $\Theta(n^2)$, as can be seen by an argument similar to the argument for
the shortest certifying walk.
\qed\end{proof}

Next, we describe an algorithm 
that is given $G$, a node $v \in V$, and a set $X \subseteq \U_{|V}$ and
then decides whether $(v,X)$ is a node of the unique sink or not.
For $W \subseteq V$, let $\U_{|W}$ denote the elements that appear in
the labels of the nodes in $W$. 
For $U \subseteq \U$, let $V_{|U}$ denote the nodes of $G$ whose
labels contain an element of $U$. 

\begin{theorem}\label{thm:decide}
Given an event graph $G$, a node $v$ of $G$ and a subset
$X \subseteq \U_{|V}$, we can decide in $O(|V|+|E|)$ steps
whether $(v,X)$ is a node of the unique sink component $\C$
of $\dec(G)$.
\end{theorem}

\begin{proof}
The idea of the algorithm is to construct a certifying walk 
for $(v, X)$ through a modified breadth first search. 

In the preprocessing phase, we color a vertex $w$ of $G$ \emph{blue} if $w$ is 
labeled \texttt{i}$x$ and $x \in X$, or if $w$ is labeled
\texttt{d}$x$ and $x \not\in X$. Otherwise, we color $w$ \emph{red}. 
If $v$ is colored red, then $(v, X)$ cannot be in $\C$, and we are done. 
Otherwise, we perform 
a directed breadth first search that starts from $v$ and tries to
construct a reverse certifying walk.
Our algorithm maintains several queues.
The main queue is called the \emph{blue fringe} $B$. Furthermore,
for every $x \in \U_{|V}$, we have a queue $R_x$, the \emph{red fringe}
for $x$. At the beginning, the queue $B$ contains only $v$, and all
the red fringes are empty.

The main loop of the algorithm takes place while $B$ is not
empty. We pull the next node $w$ out of $B$, and we process
$w$ as follows: if we have not seen the element $x_w \in \U_{|V}$
for $w$ before, we color the set $V_{|\{x_w\}}$ of all nodes 
whose label refers to $x_w$ blue, append all the nodes of 
$R_{x_w}$ to $B$, and we delete $R_{x_w}$. Next, 
we process the neighbors of $w$ as follows: if a neighbor $w'$ of $w$ is blue,
we append it to $B$ if $w'$ has not been inserted into $B$ before. If $w'$ 
is red and labeled with the element $x_{w'}$, we append $w'$ to $R_{x_{w'}}$, 
if necessary, see Fig.~\ref{fig:bfs_ex}. 

\begin{figure}
\begin{center}
\includegraphics[scale=0.9]{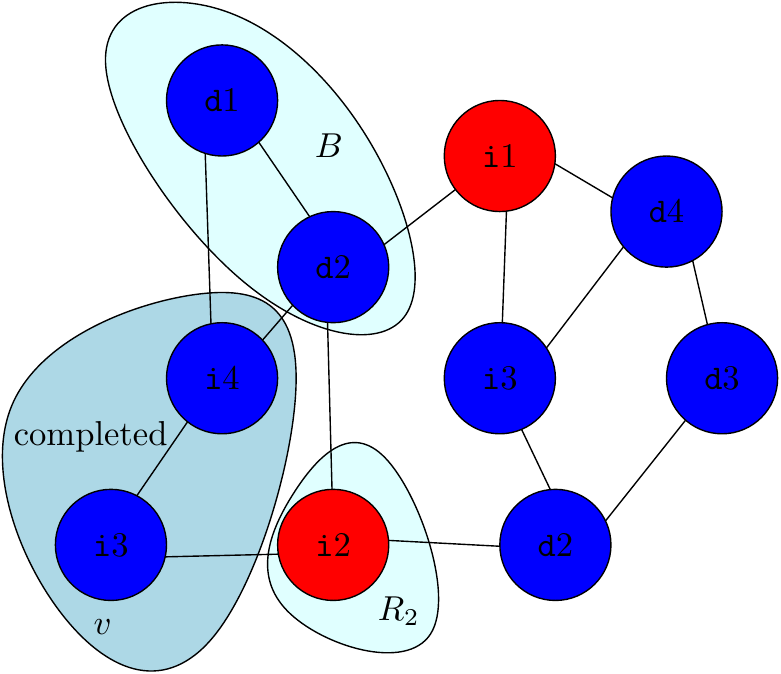}
\end{center}
\caption{An intermediate stage of the algorithm while deciding whether the 
  node $(v,\{3,4\})$ lies in the unique sink of the given event graph. At 
  this point, the nodes $v$ and $\texttt{i}4$ have been processed.
  Since the elements $3$ and $4$ have been encountered, the corresponding
  nodes have been colored blue. The nodes for the other elemenents still
  have the original color. We have $B = \{\texttt{d}1,
  \texttt{d}2\}$, $R_2 = \{\texttt{i}2\}$, and 
  $R_1 = R_3 =  R_4 = \emptyset$.  Suppose that in the next step,
  the algorithm processes $\texttt{d}2$. Then the node $\texttt{i}2$
  is colored blue and added to $B$, and $\texttt{i}1$ is added to
  $R_1$.
}
\label{fig:bfs_ex}
\end{figure}

The algorithm terminates after at
most $|V|$ iterations. In each iteration, the cost is proportional 
to the degree of the current vertex $w$ and (possibly) the size of 
one red fringe. The latter
cost can be charged to later rounds, since the nodes of the red fringe 
are processed later on. Let $\Vred$ be the union of the
remaining red fringes after the algorithm terminates.

If $\Vred = \emptyset$, 
we obtain a certifying walk for $(v, X)$ 
by walking from one newly discovered vertex to the next inside the current
blue component and reversing the walk. Now suppose $\Vred \neq \emptyset$. 
Let $A$ be the set of all vertices that were traversed during the BFS. Then
$G \setminus \Vred$ has  at least two connected
components (since there must be blue vertices outside of $A$). Furthermore, 
$\U_{|A} \cap \U_{|\Vred} = \emptyset$.
We claim that a certifying walk for $(v, X)$ cannot exist. Indeed, suppose
that $W$ is such a certifying walk. Let $x_w \in \U_{|\Vred}$ be the element
in the label of the last node $w$ in $W$ whose label refers 
to an element in $\U_{|\Vred}$. Suppose that the label of $w$ is of the form
\texttt{i}$x_w$; the other case is symmetric. Since $W$ is a
certifying walk, we have $x_w \in X$, so $w$ was colored blue during the
initialization phase. Furthermore,
all the nodes on $W$ that come after $w$ are also blue at the end. 
This implies that 
$w \in A$, because by assumption a neighor of $w$ was in $B$, and hence
$w$ must have been added to $B$ when this neighbor was processed. 
Hence, we get a contradiction to the fact that 
$\U_{|A} \cap \U_{|\Vred} = \emptyset$, so
$W$ cannot exist. Therefore, $(v,X) \not\in \C$.
\qed\end{proof}

The proof of Theorem~\ref{thm:decide} gives an alternative characterization of
whether a node appears in the unique sink component or not.
\begin{corol}
The node $(v, X)$ does not appear in $\C$ if and only if 
there exists a set $A \subseteq V(G)$ with the following properties: 
\begin{enumerate}
\item $G \backslash A$ has at least two connected components.
\item $\U_{|A} \cap \U_{|B} = \emptyset$, where $B$ denotes the
vertex set of the connected component
of $G \setminus A$ that contains $v$.
\item For all $x \in \U$, $A$ contains either only labels of the form 
\texttt{i}$x$ or only labels of the form \texttt{d}$x$ (or neither). 
If $A$ has a node with label \texttt{i}$x$, then $x \not\in X$.
If $A$ has a node with label \texttt{d}$x$, 
then $x \in X$.
\end{enumerate}
A set $A$ with the above properties can be found in polynomial time.
\qed
\end{corol}

\begin{lemma}
Given  $k \in \N$ and a node $(v, X) \in \C$, it is \textup{NP}-complete to 
decide whether there exists a certifying walk for $(v,X)$ of length at most $k$.
\end{lemma}

\begin{proof}
The problem is clearly in NP.
To show completeness, we reduce from  Hamiltonian path
in undirected graphs. Let $G$ be an undirected graph with $n$ vertices,
and suppose the vertex set is $\{1, \dots, n\}$. We let $\U = \N$ and take
two copies $G_1$ and $G_2$ of $G$. We label the copy of node $i$ in 
$G_1$ with \texttt{i}$i$ and the copy of node $i$ in $G_2$ with 
\texttt{d}$i$. Then we add two nodes $v_1$ and $v_2$, and we connect
 $v_1$ to $v_2$ and 
to all nodes in $G_1$ and $G_2$, We label $v_1$ with \texttt{i}$(n+1)$
and $v_2$ with $\texttt{d}(n+1)$. The resulting graph $G'$ has $2n+2$
nodes and meets all our assumptions about an event graph.
Clearly, $G'$ can be constructed in polynomial time.
Finally, since by definition a certifying walk must visit for
each element $i$ either $\texttt{i}i$ or $\texttt{d}i$, it follows that
$G$ has a Hamiltonian path if and only if the
node $(v_1, \{1, \dots, n+1\})$ has a certifying walk of length
at most $n+2$. This completes the reduction.
\qed\end{proof}

\section{Successor Searching on Cycle Graphs}\label{1D}

We now consider the case that the 
event graph $G$ is a simple cycle $v_1,\dots, v_{n},v_1$ and 
the item $x_{v_i}$ at node $v_i$ is a real number.
Again, the structure $\mathcal{T}(X)$ is $X$ itself,
and we now have three types of nodes: 
insertion, deletion, and query.
A query at time $t$ asks for 
$\texttt{succ}_{\,X^t}(x_{v^t})= \min\{\,x\in X^t\,|\, x\geq x_{v^t}\,\}$
(or $\infty$).
Again, an example similar to Fig.~\ref{fig:large-diameter}
shows that the decorated graph can be
of exponential size: let $n$ be even. For $i = 1, \dots, n/2 $, 
take $x_{v_i} = x_{v_{n+1-i}} = i$,
and define the operation at $v_i$ as
$\texttt{i}x_{v_i}$ for $i = 1, \dots, n/2$,
and $\texttt{d}x_{v_{n+1-i}}$ for $i = n/2 +1, \dots, n$.
It is easy to design a walk
that produces \emph{any} subset
of $\{1, \dots, n/2\}$ at either $v_1$ or $v_n$,
which 
implies a lower bound of $\Omega( 2^{n/2} )$ on
the size of the decorated graph.

We consider two different walks on $G$.
The \emph{random} walk starts at $v_1$ and hops from a node to 
one of its neighbors with equal probability.
The main result of this section is that for random walks,  maximal compression
is possible.

\begin{theorem}\label{markov-cycle}
Successor searching in a one-dimensional random
walk can be done in constant expected time per step
and linear storage.
\end{theorem}

First, however, we consider an \emph{adversarial} walk on $G$.
Note that we can always achieve a running time of $O(\log\log n)$ per step
by maintaining a van Emde Boas search structure 
dynamically~\cite{vEmdeBoasKaZi76,vEmdeBoas77}, 
so the interesting question is how little storage
we need if we are to perform each operation in constant time.

\begin{theorem}\label{adversarial-cycle}
Successor searching along an $n$-node cycle in the adversarial model
can be performed in constant time per operation, using
$O(n^{1+\varepsilon})$ storage, for any fixed $\varepsilon >0$.
\end{theorem}

Before addressing the walk on $G$, we must consider
the following range searching problem (see also~\cite{CrochemoreIlKuRaWa12}).
Let $Y= y_1,\ldots, y_n$ be a sequence of $n$ distinct numbers, and
consider the points $(k, y_k)$,  for $k = 1, \dots, n$.
A query is given by two indices $i$ and $j$, together with a \emph{type}.
The type is defined as follows:  
the horizontal lines $x \mapsto y_i$ and $x \mapsto y_j$ divide the
plane into three unbounded open strips $R_1$, $R_2$, and $R_3$, numbered from
top to bottom.  For $a = 1,2,3$, let 
$S_a = \{k \in \{1,\dots, n\} \mid (k, y_k) \text{ lies inside } R_a\}$.
The type is specified by the number $a$ together with
a direction $\rightarrow$ or $\leftarrow$. The former is called 
a \emph{right} query, the latter a \emph{left} query.
Let us describe the right query: if $S_a = \emptyset$, the 
result is $\emptyset$.  If $S_a$ contains an index larger than $i$, 
we want the minimum index in $S_a$ larger than $i$.  If all indices in $S_a$ 
are less than $i$, we want the overall minimum index in $S_a$.
The left query is defined symmetrically.
See Fig.~\ref{fig:fig12}(left) for an example.

Thus, there are six types of queries, and we specify a 
query by a triplet $(i,j,\sigma)$, with $\sigma$ to being the type. 
We need the following result, which, as a reviewer
pointed out to us, was also discovered earlier by 
Crochemore~et~al.\cite{CrochemoreIlKuRaWa12}.
We include our proof below for completeness.

\begin{figure}
\begin{center}
\includegraphics{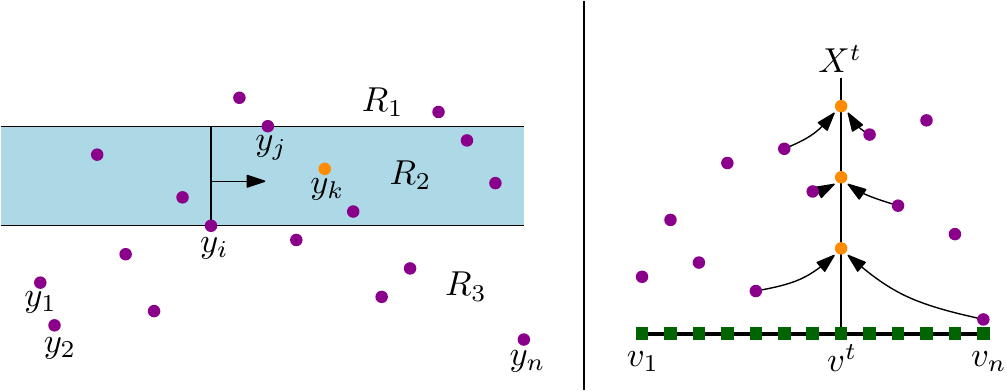}
\end{center}
\caption{Left: the query  $(i,j, (2,\rightarrow))$, we want the leftmost
  point to the right of $y_i$ in the strip $R_2$;
 right: 
the successor data structure. The squares at the bottom 
represent the vertices of the cycle, split at the edge $v_nv_1$
to obtain a better picture. The dots above the cycle nodes represent
the elements $x_{v_i}$. 
The node $v^t$ is the current node, and $X^t$ the active set. We maintain
pointers between each element $x \in X^t$ and the closest clockwise and
counterclockwise node such that the successor in $X^t$ of the corresponding
element is $x$.}
\label{fig:fig12}
\end{figure}

\begin{lemma}\label{range-search}
Any query can be answered in constant time
with the help of a data structure of size $O(n^{1+\varepsilon})$,
for any $\varepsilon >0$.
\end{lemma}

\noindent
Using Lemma~\ref{range-search}, we can prove 
Theorem~\ref{adversarial-cycle}.
\begin{proof}[of Theorem~\ref{adversarial-cycle}]
At any time $t$, the algorithm has at its disposal:
(i) a sorted doubly-linked list of the active set $X^t$ 
(augmented with $\infty$);
(ii) a (bidirectional) pointer to each $x\in X^t$ from
the first node $v_k$ on the circle clockwise from $v^t$, if
it exists, such that ${\tt succ}_{\,X^t}(x_{v_k})= x$
(same thing counterclockwise)---see Fig.~\ref{fig:fig12}(right).
Assume now that the data structure of Lemma~\ref{range-search} has been
set up over $Y =  x_{v_1},\ldots, x_{v_n}$.
As the walk enters node $v^t$ at time $t$, ${\tt succ}_{\,X^t}(x_{v^t})$
is thus readily available and we can update $X^t$ in $O(1)$ time.
The only remaining question is how to maintain (ii). 
Suppose that the operation at node $v^t$ is a successor request
and that the walk reached $v^t$ clockwise.
If $x$ is the successor, then 
we need to find the first node $v_k$ on the cycle clockwise from $v^t$ 
such that ${\tt succ}_{\,X^t}(x_{v_k})= x$.
This can be handled by two range search queries $(i,j,\sigma)$:
for $i$, use the index of the current node $v^t$; and, for $j$, 
use the node for $x$ in the first query and the node for $x$'s 
predecessor in $X^t$ in the second query.
An insert can be handled by two such queries (one on each side of $v_t$),
while a delete requires pointer updating, but no range search queries.
\qed\end{proof}

\begin{proof}[of Lemma~\ref{range-search}]
We define a single data structure to handle all six types
simultaneously. We restrict our discussion to the type $(2, \rightarrow)$ from
Fig.~\ref{fig:fig12}(left) but kindly invite the reader to 
check that all other five types can be handled in much the same way.
We prove by induction that with $s c n^{1+1/s}$ storage, for a
large enough constant $c$, any query can be answered in at most $O(s)$ 
table lookups.
The case $s=1$ being obvious (precompute all queries),
we assume that $s>1$.
Sort and partition $Y$ into consecutive groups $Y_1<\cdots < Y_{n^{1/s}}$
of size $n^{1-1/s}$ each. We have two sets of tables:

\begin{itemize}
\item
\textbf{Ylinks}: for each $y_i\in Y$, 
link $y_i$ to the highest-indexed element $y_j$ to the left of $i$ ($j<i$) 
within each group $Y_{1},\ldots, Y_{n^{1/s}}$, wrapping around the
strip if necessary
(left pointers in Fig.~\ref{fig:fig34}(left)).
\item
\textbf{Zlinks}: for each $y_i\in Y$, 
find the group $Y_{\ell_i}$ to which $y_i$ belongs and,
for each $k$, define $Z_k$ as the subset of $Y$
sandwiched  between $y_i$ and the smallest (resp.~largest) 
element in $Y_k$ if $k\leq \ell_i$ (resp.  $k\geq \ell_i$).
Note that this actually defines two sets for $Z_{\ell_i}$, so that
the total number of $Z_k$'s is really $n^{1/s}+1$.
Link $y_i$ to the lowest-indexed $y_j$ ($j>i$) in each $Z_k$
(right pointers in Fig.~\ref{fig:fig34}(left)), again
wrapping around if necessary.

\item
Prepare a data structure of type $s-1$ recursively for each $Y_i$.
\end{itemize}

Given a query $(i,j)$ of type $(2, \rightarrow)$, we first 
check whether it
fits entirely within $Y_{\ell_i}$ and, if so, solve it recursively.
Otherwise, we break it down into two subqueries: 
one of them can be handled directly by using the relevant Zlink.
The other one fits entirely within a single $Y_k$.
By following the corresponding Ylink, we find $y_{i'}$
and solve the subquery recursively by converting it
into another query $(i',j)$ of appropriate type
(Fig.~\ref{fig:fig34}(right)).
By induction, it follows that this takes $O(s)$ total lookups and storage 
\[
dn^{1+1/s} + (s-1) c n^{1/s + (1-1/s)(1+1/(s-1))} =
dn^{1+1/s} + (s-1) c n^{1 + 1/s} \leq 
s c n^{1+1/s}, 
\]
for some constant $d$ and for $c$ large enough, since
\[
\Bigl(1-\frac{1}{s}\Bigr)\Bigl(1+\frac{1}{s-1}\Bigr) = 
\frac{s-1}{s} \frac{s}{s-1} = 1.
\]
\qed\end{proof}

\begin{figure}
\begin{center}
\includegraphics{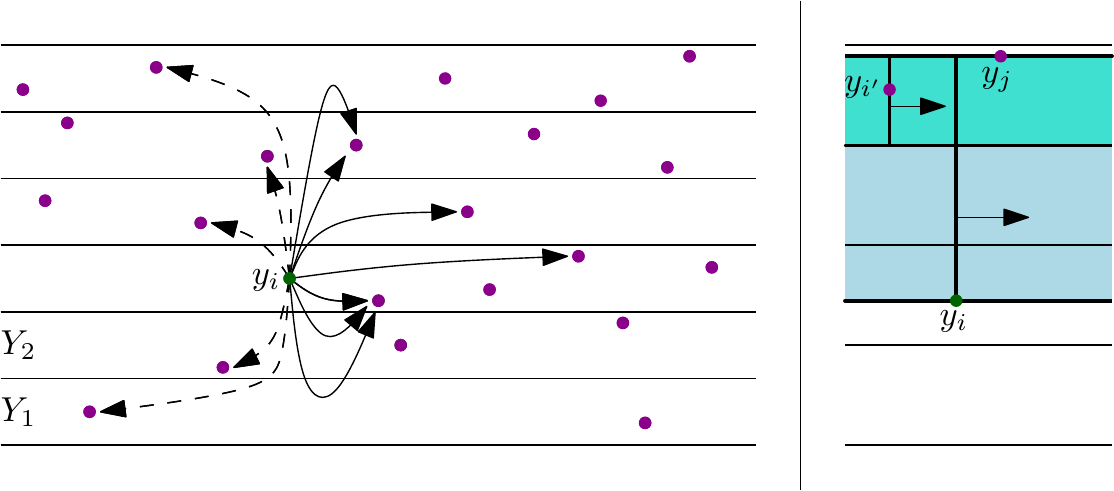}
\end{center}
\caption{Left: the recursive data structure: The Ylinks (dashed) point
to the rightmost point to the left of $y_i$ in each strip.
The ZLinks point to the leftmost point in each block defined by
$y_i$ and a consecutive sequence of strips; right:
a query $(i,j)$ is decomposed into a part handled by
a ZLink and a part that is handled recursively. }
\label{fig:fig34}
\end{figure}

Using Theorem~\ref{adversarial-cycle} together
with the special properties of a random walk on $G$,
we can quickly derive the algorithm for Theorem~\ref{markov-cycle}.

\begin{proof}[of Theorem~\ref{markov-cycle}]
The idea is to divide up the cycle
into $\sqrt{n}$ equal-size paths $P_1,\ldots, P_{\sqrt{n}}$
and prepare an adversarial data structure for 
each one of them right upon entry. The high
cover time of a one-dimensional random walk 
is then invoked to amortize the costs.
De-amortization techniques are then used to make the costs worst-case.
The details follow. As soon as the walk enters
a new $P_k$, the data structure of Lemma~\ref{range-search}
is built from scratch for $\varepsilon=1/3$, at a cost in
time and storage of $O(n^{2/3})$.
By merging $L_k= \{\,x_{v_i}\,|\, v_i\in P_k\, \}$ with the doubly-linked list
storing $X^t$, we can set up all the needed successor links
and proceeds just as in Theorem~\ref{adversarial-cycle}.
This takes $O(n)$ time per interpath transition
and requires $O(n^{2/3})$ storage.
There are few technical difficulties that we now
address one by one.

\begin{figure}[ht]
\begin{center}
\includegraphics{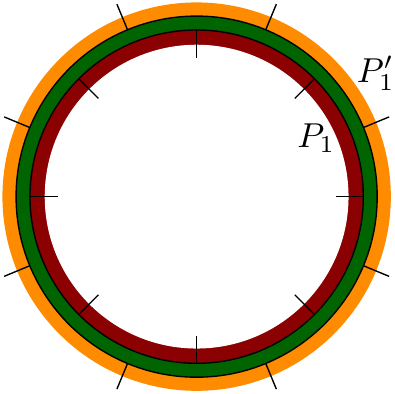}
\end{center}
\caption{The parallel tracks on the cycle.}
\label{fig:fig5}
\end{figure}

\begin{itemize}
\item
Upon entry into a new path $P_k$, we must set up 
successor links from $P_k$ to $X^t$, which takes $O(n)$ time.
Rather than forcing the walk to a halt, we use a ``parallel track'' idea
to de-amortize these costs. (Fig.~\ref{fig:fig5}).
Cover the cycle with paths $P_i'$ shifted from
$P_i$ clockwise by $\frac{1}{2}\sqrt{n}$.
and carry on the updates in parallel on both tracks.
As we shall see below, we can ensure that
updates do not take place simultaneously on both tracks.
Therefore, one of them is always
available to answer successor requests in constant time.
\item
Upon entry into a new path $P_k$ (or $P_k'$),
the relevant range search structure must be built from scratch.
This work does not require knowledge of $X^t$
and, in fact, the only reason
it is not done in preprocessing is to save storage.
Again, to avoid having to interrupt the walk,
while in $P_k$ we ensure that the needed
structures for the two adjacent paths $P_{k-1},P_{k+1}$ are already available
and those for $P_{k-2},P_{k+2}$ are under construction.
(Same with $P_k'$.) 
\item
On a path, we do not want our range queries to wrap around as 
in the original structure. Thus, if a right query returns an index
smaller than $i$, or a left query returns an index larger than $i$,
we change the answer to $\emptyset$.
\item
The range search structure can only handle queries $(i,j)$ for which
{\em both} $y_i$ and $y_j$ are in the ground set. Unfortunately, $j$ may not be,
for it may correspond to an item of $X^t$ inserted prior
to entry into the current $P_k$.
There is an easy fix: upon entering $P_k$, compute and store
${\tt succ}_{\,L_k}(x_{v_i})$ for $i=1,\ldots, n$.
Then, simply replace a query $(i,j)$ by $(i,j')$ where $j'$
is the successor (or predecessor) in $L_k$.
\end{itemize}

The key idea now is that a one-dimensional random walk has 
a quadratic cover time~\cite{MotwaniRa95};
therefore, the expected
time between any change of paths on one track and
the next change of paths on the other track is $\Theta(n)$.
This means that if we dovetail the parallel updates by
performing a large enough number of them per walking step,
we can keep the expected time per operation constant.
This proves Theorem~\ref{markov-cycle}.
\qed\end{proof}

\section{Conclusion}

We have presented a new approach to model and analyze restricted
query sequences that is inspired by Markov chains. 
Our results only scratch the surface of a rich body of questions.
For example, even for the simple problem of the adversarial walk
on a path, we still do not know whether we can beat van Emde
Boas trees with linear space. Even though there is some 
evidence that the known lower bounds for successor searching
on a pointer machine
give the adversary a lot of leeway~\cite{Mulzer09}, our
lower bound technology
does not seem to be advanced enough for this setting.
Beyond paths and cycles, of course, there are several other
simple graph classes to be explored, e.g., trees or planar
graphs.

Furthermore, there are more fundamental questions on decorated
graphs to be studied. For example, how hard is it to count
the number of distinct active sets (or the number of nodes) 
that occur in the unique sink component of $\dec(G)$? 
What can we say about the  behaviour of the active
set in the limit as the walk proceeds randomly? And what happens
if we go beyond the dictionary problem and consider the evolution
of more complex structures during a walk on the event graph?

\section*{Acknowledgments}
We would like to thank the anonymous referees for their thorough
reading of the paper and their many helpful suggestions that have
improved the presentation of this paper, as well
as for pointing out~\cite{CrochemoreIlKuRaWa12} to us.

W. Mulzer was supported in part by DFG grant MU3501/1.

\bibliographystyle{abbrv}
\bibliography{dec}

\newcommand{\SortNoop}[1]{}
\begin{thebibliography}{10}

\bibitem{Chassaing93}
P.~Chassaing.
\newblock Optimality of move-to-front for self-organizing data structures with
  locality of references.
\newblock {\em The Annals of Applied Probability}, 3(4):1219--1240, 1993.

\bibitem{Chazelle00}
B.~Chazelle.
\newblock {\em The discrepancy method: randomness and complexity}.
\newblock Cambridge University Press, Cambridge, 2000.

\bibitem{ChazelleMu09}
B.~Chazelle and W.~Mulzer.
\newblock Markov incremental constructions.
\newblock {\em Discrete \& Computational Geometry}, 42(3):399--420, 2009.

\bibitem{CrochemoreIlKuRaWa12}
M.~Crochemore, C.~S. Iliopoulos, M.~Kubica, M.~S. Rahman, G.~Tischler, and
  T.~Wale.
\newblock Improved algorithms for the range next value problem and
  applications.
\newblock {\em Theoretical Computer Science}, 434:23--34, 2012.

\bibitem{vEmdeBoas77}
P.~{\SortNoop{Emde Boas}}van Emde~Boas.
\newblock Preserving order in a forest in less than logarithmic time and linear
  space.
\newblock {\em Information Processing Letters}, 6(3):80--82, 1977.

\bibitem{vEmdeBoasKaZi76}
P.~{\SortNoop{Emde Boas}}van Emde~Boas, R.~Kaas, and E.~Zijlstra.
\newblock Design and implementation of an efficient priority queue.
\newblock {\em Mathematical Systems Theory}, 10(2):99--127, 1976.

\bibitem{Hotz93}
G.~Hotz.
\newblock Search trees and search graphs for {M}arkov sources.
\newblock {\em Elektronische Informationsverarbeitung und Kybernetik},
  29(5):283--292, 1993.

\bibitem{KapoorRe91}
S.~Kapoor and E.~M. Reingold.
\newblock Stochastic rearrangement rules for self-organizing data structures.
\newblock {\em Algorithmica}, 6(2):278--291, 1991.

\bibitem{KarlinPhRa00}
A.~R. Karlin, S.~J. Phillips, and P.~Raghavan.
\newblock {M}arkov paging.
\newblock {\em SIAM Journal on Computing}, 30(3):906--922, 2000.

\bibitem{KonnekerVa81}
L.~K. Konneker and Y.~L. Varol.
\newblock A note on heuristics for dynamic organization of data structures.
\newblock {\em Information Processing Letters}, 12(5):213--216, 1981.

\bibitem{LamLeSi84}
K.~Lam, M.~Y. Leung, and M.~K. Siu.
\newblock Self-organizing files with dependent accesses.
\newblock {\em Journal of Applied Probability}, 21(2):343--359, 1984.

\bibitem{LevinPeWi09}
D.~A. Levin, Y.~Peres, and E.~L. Wilmer.
\newblock {\em Markov chains and mixing times}.
\newblock American Mathematical Society, Providence, RI, 2009.

\bibitem{MotwaniRa95}
R.~Motwani and P.~Raghavan.
\newblock {\em Randomized algorithms}.
\newblock Cambridge University Press, Cambridge, 1995.

\bibitem{Mulzer09}
W.~Mulzer.
\newblock A note on predecessor searching in the pointer machine model.
\newblock {\em Information Processing Letters}, 109(13):726--729, 2009.

\bibitem{PhatarfodPrDy97}
R.~M. Phatarfod, A.~J. Pryde, and D.~Dyte.
\newblock On the move-to-front scheme with {M}arkov dependent requests.
\newblock {\em Journal of Applied Probability}, 34(3):790--794, 1997.

\bibitem{SchulzSc96}
F.~Schulz and E.~Sch{\"o}mer.
\newblock Self-organizing data structures with dependent accesses.
\newblock In {\em Proceedings of the 23rd International Colloquium on Automata,
  Languages, and Programming (ICALP)}, pages 526--537, 1996.

\bibitem{ShedlerTu72}
G.~S. Shedler and C.~Tung.
\newblock Locality in page reference strings.
\newblock {\em SIAM Journal on Computing}, 1(3):218--241, 1972.

\bibitem{VitterKr96}
J.~S. Vitter and P.~Krishnan.
\newblock Optimal prefetching via data compression.
\newblock {\em Journal of the ACM}, 43(5):771--793, 1996.

\end{thebibliography}

\end{document}